\documentclass[11pt]{article}

\usepackage{amsmath, amsthm, amssymb}
\usepackage{graphicx}
\usepackage{epstopdf}
\usepackage{psfrag}
\usepackage{latexsym}
\usepackage{amsfonts}
\usepackage{url}
\usepackage{algorithm}
\usepackage{algpseudocode}
\usepackage{calc}
\usepackage{tikz}
\usepackage{pgfplots}
\usepackage{fullpage}
\usepackage{hyperref}

\algnewcommand\ExitFor{\textbf{exit for}}

\newtheorem{thm}{Theorem}[section]
\newtheorem{prop}[thm]{Proposition}
\newtheorem{conjecture}[thm]{Conjecture}
\newtheorem{defi}[thm]{Definition}
\newtheorem{lemma}[thm]{Lemma}

\newtheorem{exm}[thm]{Example}

\newcommand{\C} {\mathbb C}

\newcommand{\K}{\mathbb K}
\newcommand{\Lex}{Lex}
\newcommand{\NP}{\text{NP}}
\newcommand{\xx}[0]{\textbf{x}}
\newcommand{\z}[0]{\zeta}
\newcommand{\lz}[0]{\langle}
\newcommand{\rz}[0]{\rangle}
\newcommand{\A}[0]{\mathcal{A}}
\newcommand{\F}[0]{\mathcal{F}}
\newcommand{\G}[0]{\mathcal{G}}
\newcommand{\V}[0]{\mathcal{V}}
\newcommand{\I}[0]{\mathcal{I}}

\def \C{\mathbb{C}}

\def \supp {\mbox{supp}}
\def \rad {\mbox{rad}}

\def \LM {\mbox{LM}}
\def \len {\mbox{len}}

\begin{document}
\bibliographystyle{alpha}
\unitlength=1cm

\title{Gr\"obner Bases and Nullstellens\"atze for Graph-Coloring Ideals}

\author{Jes\'us A. De Loera\thanks{University of California, Davis, \url{deloera@math.ucdavis.edu}}, Susan Margulies\thanks{United States Naval Academy, \url{margulie@usna.edu}}, \\Michael Pernpeintner\thanks{Technische Universit\"at M\"unchen, \url{michaelpernpeintner@gmail.com}}, Eric Riedl\thanks{Harvard University, \url{ebriedl@math.harvard.edu}}, David Rolnick\thanks{Massachusetts Institute of Technology, \url{drolnick@math.mit.edu}}, \\Gwen Spencer\thanks{Smith College, \url{gspencer@smith.edu}}, Despina Stasi\thanks{Illinois Institute of Technology, \url{stasdes@iit.edu}}, Jon Swenson\thanks{University of Washington, \url{jmswen@math.washington.edu}}}

\maketitle

\begin{abstract}
We revisit a well-known family of polynomial ideals encoding the problem of graph-$k$-colorability. Our paper describes how the inherent combinatorial structure of the ideals implies several interesting algebraic properties.  Specifically, we provide lower bounds on the difficulty of computing Gr\"obner bases and Nullstellensatz certificates for the coloring ideals of general graphs.  For chordal graphs, however, we explicitly describe a Gr\"obner basis for the coloring ideal, and provide a polynomial-time algorithm.
\end{abstract}
%

\section{Introduction}

Many authors in computer algebra and complexity theory have studied the complexity of Gr\"obner bases (see e.g., \cite{dube,MayrMeyer,MayrRitscher,ritscher,GG} and references therein) and
the difficulty of Hilbert's Nullstellensatz (see \cite{beamer,brownawell, buss,D'Andreaetal,kollar,lazard,Mayrsurvey}). With few exceptions authors have concentrated 
on proving worst-case upper bounds. 
In this paper we look at the behavior of Gr\"obner bases and Hilbert Nullstellens\"atze in a
combinatorial family of polynomials. Our key point is to study how the structure of graph coloring problems provides lower bounds on the difficulty of finding Gr\"obner bases and Nullstellensatz certificates, providing a counterpart to upper bound theorems for general polynomial systems.

Many authors have studied the rich connection between graphs and polynomials (see e.g., \cite{alonsurvey,alontarsi,deloera,matiyasevich2,Mnuk,lovasz1}
Here our starting point is Bayer's theorem for $3$-colorings \cite{bayer}, further generalized in \cite{DLMM_issac, DLMO} to $k$-coloring over a finite field: \\[5pt]
Suppose we wish to check whether a graph $G=(V, E)$ is $k$-colorable, and set $n=|V|$.  We consider the \emph{$k$-coloring ideal} $\I_k(G) \subset \mathbb C[x_1,\ldots, x_n]$ (also denoted by $\I_G$ if the number of colors is clear) generated by the {\em vertex polynomials} $\nu_i:=x_i^k-1$, for $1\leq i\leq n$, and the {\em edge polynomials} $\eta_{i, j}:=\sum_{l=0}^{k-1} x_i^lx_j^{k-1-l}$, for $\{i, j\}\in E$. The set of all vertex and edge polynomials of a graph $G$ is denoted by $\F_G$.

\begin{thm}[\cite{DLMM_issac, DLMO}]
The graph $G$ is $k$-colorable if and only if $\I_k(G)$ has a common root. In other words, $G$ is not $k$-colorable if and only if $\I_k(G)=\lz1\rz=\C[x_1,\dots,x_n]$. Moreover, the dimension of the vector space $\C[x_1,\dots,x_n]/\I_k(G)$ equals $k!$ times the number of distinct $k$-colorings of $G$.
\end{thm}

From the well-known Hilbert's Nullstellensatz \cite{coxetal}, one can derive {\em certificates} that a system of polynomials has no solution (i.e., in our case, that a graph does not have a $k$-coloring).

\begin{thm}[Hilbert's Nullstellensatz \cite{coxetal}]
Suppose that $f_1,\ldots,f_m\in \K[x_1,\ldots,x_n]$.  Then, there are no solutions to the system $\{f_i=0\}$ in the algebraic closure of $\K$, if and only if there exist $\alpha_i\in \K[x_1,\ldots,x_n]$ such that $$\alpha_1f_1+\cdots+\alpha_mf_m=1.$$
\end{thm}

We refer to the set $\{\alpha_i\}$ as a \emph{Nullstellensatz certificate}, and measure the complexity of a certificate by its \emph{degree}, defined as the maximum degree of any $\alpha_i$.  If a system is known to have a Nullstellensatz certificate of small constant degree (over a finite field), one can simply find this certificate by a series of linear algebra computations \cite{DLMM_issac,DLMM_jsc,malkinetal}. There are well-known \emph{upper bounds} for the degrees of the coefficients $\alpha_i$ in the Nullstellensatz certificate for \emph{general} systems of polynomials that grow with the number of variables \cite{kollar}. Furthermore, these bounds turn out to be sharp for some pathological instances. 

 
Connections between complexity theory and the Gr\"obner bases and Nullstellens\"atze of the
coloring ideals have been made in \cite{beamer, lovasz1, susan_thesis}: It is known that unless $\text{NP}=\text{coNP}$, there must exist an infinite family of non-3-colorable graphs for which the minimum degree of a Hilbert Nullstellensatz certificate grows arbitrarily large \cite{DLMM_jsc, susan_thesis}.  In upcoming work, Cifuentes and Parrilo \cite{parrilo} identify graph structure within an arbitrary polynomial system and show that this yields faster algorithms for solving systems of polynomials.   For further background on the material presented in this paper, we direct the interested reader to the books \cite{adamsloustaunau,coxetal,coxetalII,gareyjohnson}.

This article offers three new contributions in the complexity of working with coloring ideals.

(1) In Section \ref{certificates}, we show that the minimal degree of Nullstellensatz certificates of coloring ideals satisfies certain modular constraints and grows at least linearly in the number of colors. We indicate that the field of coefficients has some intriguing influence in the complexity and propose a conjecture.

(2) Recall that an algorithm $\A$ is an {\em $\alpha$-approximation algorithm} if, for every input instance of the (minimization) problem, $\A$ delivers a feasible solution of cost no more than $\alpha$ times the optimal possible cost in polynomial time.  It is well-known that many combinatorial problems cannot be well-approximated. For instance, Khanna et al. \cite{khanna} have shown it is $\text{NP}$-hard to 4-color a 3-colorable graph. More strongly, even if one is allowed to ignore a particular small (but non-zero) fraction of nodes, it is $\text{NP}$-hard to properly 4-color the remaining nodes.

In Section \ref{extrahard}, we demonstrate how one can transfer inapproximability results for graphs to inapproximability results for polynomial rings. We prove that it is hard to compute a Gr\"obner basis for an ideal even if we are allowed to ignore a large subset of the generators for our ideal.  This shows how the coloring ideal provides a sense of ``robust'' hardness for the computation of Gr\"obner bases.

(3) Despite hardness in the general case of computing a Gr\"obner basis, we might hope that some algorithm could find Gr\"{o}bner bases efficiently, particularly if we restricted our focus to some special class of relatively simple systems of polynomials. In Section \ref{sec:chordal}, we prove that computing a Gr\"obner basis can be done in polynomial time when the associated graph is chordal, and we describe explicitly the structure of such a Gr\"obner basis.





\section{Lower Bounds on Hardness: Gr\"obner Bases \& Nullstellens\"atze}
\label{ExpectedHardness}

For general ideals $I\subseteq\mathbb{K}[x_1, \ldots, x_n]$, it is well-known that the computation of Gr\"obner bases is $\NP$-hard. This follows directly because many $\NP$-hard problems can be easily encoded as the solution of a multivariate polynomial system (see e.g., \cite{bayer,DLMO,deloera}). In particular, it is obvious that if the system of equations in the coloring ideal can be solved in polynomial time (in the input size) for 3-coloring ideals, then $\text{P}=\NP$. What makes this
very interesting is that one can see (or at least try to see) algebraic phenomena that are produced by the separation of complexity classes. For example, assuming that $\text{P} \not= \NP$ then the degree of Nullstellensatz certificates for systems of equations coming from non-3-colorable graphs must show some growth in the degree. In what follows, we discuss two ways in which the hardness of solving the coloring problem algebraically is made concrete.

\subsection{Nullstellens\"atze}
\label{certificates}


In this section, we consider $N_{k,\K}(G)$, the minimal Nullstellensatz degree for the $k$-coloring ideal of a graph $G$ over the field $\K$. We show that $N_{k,\K}(G)$ grows at least linearly with respect to $k$, and provide evidence that the growth is, in fact, faster.  Note that $N_{k,\K}(G)$ is defined for all $G$ that are \emph{not} $k$-colorable, and for all fields $\K$ for which the characteristic does not divide $k$.  Our main result is the following:
 
\begin{thm}
\label{thm:degree}
For all $k,\K,G$, we have $N_{k,\K}(G)\equiv 1\pmod{k}$.  Furthermore $N_{k,\K}(G)\ge k+1$ if $k>3$.
\end{thm}

\begin{proof}
Let $G=(V,E)$, and let $\I_G$ denote the $k$-coloring ideal of $G$.  Then, $\I_G$ is generated by vertex polynomials $\nu_i=x_i^k-1$ (for $i\in V(G)$) and edge polynomials $\eta_{ij}=(x_i^k-x_j^k)/(x_i-x_j)$ (for $(i,j)\in E(G)$).  We note that $\I_G$ has a Nullstellensatz certificate over $\K[x_1,\ldots,x_n]$ if and only if it has such a certificate over $\K[x_1,\ldots,x_n]/\lz x_1^k-1,\ldots,x_n^k-1\rz$.  Therefore, we may consider only the edge polynomials $\eta_{ij}$ and assume that degrees of variables are taken modulo $k$, that is, $x_i^k=1$ for every $i$. 

Suppose that $\{\alpha_{ij}\}$ is a Nullstellensatz certificate of degree $d$, so that $\sum_{ij\in E} \alpha_{ij} \eta_{ij}=1$.  We write $\alpha_{ij}=\sum_t \alpha_{t,ij}$, where $\alpha_{t,ij}$ is homogeneous of degree $t$.  Equating terms of like degree, we see that, for each $t\not\equiv 1\pmod k$,
\begin{align*}
\sum_{ij\in E}\alpha_{t,ij}\eta_{ij}=0,\quad \text{and~} \sum_{ij\in E, t\equiv 1}\alpha_{t,ij}\eta_{ij}=1.
\end{align*}
Hence, letting $\beta_{ij}=\sum_{t\equiv 1}\alpha_{t,ij}$, observe that $\{\beta_{ij}\}$ is a Nullstellensatz certificate with degree congruent to 1 modulo $k$.  We conclude that $N_{k,\K}(G)\equiv 1\pmod{k}$.

Now consider $k>3$ and suppose towards contradiction that there exists a Nullstellensatz certificate $\{\alpha_{ij}\}$ of degree at most 1.  By our logic above, we need only consider terms in $\alpha_{ij}$ for which the degree is 1 modulo $k$.  Suppose therefore that $\alpha_{ij}=\sum_h c_{h,ij}x_h$, so that $$\sum_{h\in V,\, ij\in E} c_{h,ij}x_h \eta_{ij}=1.$$Notice that $c_{h,ij}x_h \eta_{ij}$ can contain a constant term only when $h$ equals $i$ or $j$, in which case $x_h(x_i^{k-1})$ or $x_h(x_j^{k-1})$ equals 1.  We conclude that \begin{equation}1=\sum_{ij\in E} (c_{i,ij}+c_{j,ij})\label{ij}.\end{equation}

Observe that $c_{i,ij}x_i \eta_{ij}$ contains a term of the form $c_{i,ij}x_i^{k-2}x_j^2$.  Since $k>3$, the monomial $x_i^{k-2}x_j^2$ occurs for only one other choice of $h'$ and $i'j'$, namely $i'=i$ and $h'=j'=j$.  In order for this term to cancel in the final sum, therefore, we must have $c_{j,ij}=-c_{i,ij}$ for all $ij\in E$.  However, this contradicts (\ref{ij}).  We conclude that for $k>3$, no Nullstellensatz certificate exists of degree 1, and therefore that $N_{k,\K}(G)\ge k+1$.
\end{proof}

We observe that Thm. \ref{thm:degree} is a generalization of Lemmas 4.0.48 and 4.0.49 of \cite{susan_thesis}, which only deals with the graph-3-colorability case.

\begin{exm} Consider the following \textbf{incomplete} degree four certificate for non-3-colorability over $\mathbb{F}_2$. Observe that the coefficient for the vertex polynomial $(x_0^3 + 1)$ contains only monomials of degree zero and degree three, whereas the coefficient for the edge polynomial $(x_0^2+x_0x_2+x_2^2)$ contains only monomials of degree one or degree four. This certificate demonstrates the modular degree grouping of the monomials in the certificates, as described by Thm \ref{thm:degree}. We do not display the total certificate here due to space considerations.
\begin{align*}
1 &= (1+x_0x_2x_4+x_0x_2x_6+x_0x_3x_4+x_0x_3x_5+x_0x_4x_5+x_0x_4x_6+x_1^2x_4\\
&+x_1^2x_6+x_1x_3x_4+x_1x_3x_5+x_1x_5x_6+x_2x_3x_4+x_2x_3x_6+x_3x_5x_6+x_4x_5x_6)(x_0^3+1)\\
&+(x_1+x_3+x_4+    x_0^2x_1x_4+x_0^2x_1x_6+x_0^2x_2x_4+x_0^2x_2x_6+ x_0^2x_3x_4+x_0^2x_3x_5+x_0^2x_5x_6\\
&+x_0x_1x_3x_4+x_0x_1x_3x_6+x_0x_2x_3x_4+x_0x_2x_3x_6+x_0x_2x_4x_5+x_0x_2x_4x_6+x_0x_2x_5x_6\\
&+x_0x_3x_4x_5+x_0x_3x_4x_6+x_0x_3x_5x_6+x_0x_4x_5x_6+x_1x_3x_4x_5\\
&+x_1x_3x_4x_6+x_1x_4x_5x_6+x_2x_3x_4x_5+x_2x_3x_4x_6)(x_0^2+x_0x_1+x_1^2)\\
&+(x_1+x_3+ x_4+x_6+x_0^2x_1x_4+x_0^2x_1x_6+x_0^2x_4x_5+x_0^2x_4x_6+x_0^2x_5x_6+x_0x_1x_3x_4\\
&+x_0x_1x_3x_6+x_0x_3x_4x_5+x_0x_3x_4x_6+x_1x_3x_4x_5+x_1x_3x_4x_6+x_1x_4x_5x_6\\
&+x_3x_4x_5x_6)(x_0^2+x_0x_2+x_2^2)+\mathbf{\cdots}
\end{align*}
\end{exm}

\subsubsection{Experiments and future directions}

In Table \ref{fig:degree}, we display experimental data on minimum-degree Nullstellensatz certificates for various graph-$k$-colorability cases and various finite fields. This data was found via the high-performance computing cluster at the US Naval Academy (and the NulLa software \cite{DLMM_issac}). Observe in particular that the Nullstellensatz certificate computed by testing the complete graph $K_7$ for non-6-colorability is \textbf{not} the minimum possible degree: instead of expected minimum degree seven, the minimum-degree certificate is the next higher residue class (degree thirteen). Additionally (not presented in Table \ref{fig:degree}), we tested non-3-colorability for $K_4$ for the first 1,000 prime finite fields. The certificate degree was degree one for finite fields $\mathbb{F}_2$ and $\mathbb{F}_5$, and changed to the next highest degree (degree four) at $\mathbb{F}_7$, and the remained degree four for the next 997 primes (up to $\mathbb{F}_{7919}$). We also (not presented Table \ref{fig:degree}) tested non-4-colorability for $K_5$ for the first 1,000 primes: the minimum-degree remained five for the entire series of computations. In general, Table \ref{fig:degree} suggests that the bound $N_{k,\K}(G)\ge k+1$ for $k\ge 4$ is not tight for large $k$. This yields the following conjecture:

\begin{conjecture} For every field $\K$ and for every positive integer $m$, there exists a constant $k_0$ with the following property.  For each $k>k_0$ and $G$ a non-$k$-colorable graph, every Nullstellensatz certificate of the $k$-coloring ideal of $G$ has degree at least $mk+1$.
\end{conjecture}

\begin{table}[htbp]
$$\begin{array}{c|c|c||c|c|c|c}
\text{Graph}&k&\text{Possible degrees (Theorem \ref{thm:degree})}&\mathbb{F}_2&\mathbb{F}_3&\mathbb{F}_5&\mathbb{F}_7\\ \hline\hline
K_4&3&1,4,7,10,\ldots&1&-&4&4\\ \hline
K_5&4&5,9,13,\ldots&-&5&5&5\\ \hline
K_6&5&6,11,16,\ldots&6&6&-&11\\ \hline
K_7&6&7,13,19,\ldots&-&-&13&13\\ \hline
K_8&7&8,15,22,\ldots&8&\ge 15&\ge 15&-\\ \hline
K_9&8&9,17,25,\ldots&-&\ge 17&\ge 17&\ge 17\\ \hline
K_{10}&9&10,19,28,\ldots&\ge 19&-&\ge 19&\ge 19\\ \hline
K_{11}&10&11,21,31,\ldots&-&\ge 21&-&\ge 21\\
\end{array}$$
\caption{The minimum degree of Nullstellensatz certificates for complete graphs over $\mathbb{F}_p$.  Note that computations are only possible when $k$ and $p$ are relatively prime (incompatible pairs $(k,p)$ are denoted by $-$).}
\label{fig:degree}
\end{table}

\subsection{The Extra Hardness of Colorful Gr\"obner bases}
\label{extrahard}

We know it is $\NP$-hard to compute Gr\"obner bases. It is even known the problem is EXPSPACE-complete (see \cite{GG}), and the maximum degree of the basis can become very large. An upper bound is given in \cite{MayrRitscher} for the degree of a reduced Gr\"obner basis for an $r$-dimensional ideal, whose generators have degree bounded by $d$. The authors show that a Gr\"obner basis of such an ideal can have degree $\leq 2\left(\frac{1}{2}d^{n-r}+d\right)^{2^r}$. For the case of general zero-dimensional ideals, this bound reduces to $\leq 2\left(\frac{1}{2}d^n+d\right)$. In \cite{ritscher}, a lower bound of $d^n$ for zero-dimensional ideals is given by a suitable example. Finally, Lazard and Brownawell \cite{brownawell,lazard} independently proved an $n(d-1)$ bound on specialized zero-dimensional ideals that include our coloring ideals, $\I_k(G)$.

On the other hand, it is well-known that some combinatorial problems are even hard to approximate or it is hard to find partial solutions. Here we discuss how the hardness of finding suboptimal or approximate solutions to graph $k$-coloring can be translated into similar results for the computation of Gr\"obner bases, therefore showing some kind of ``robust hardness" for Gr\"obner bases computation. We will use the following theorem.


\begin{thm}[see \cite{khanna}]\label{ApproxHard}
It is $\NP$-hard to color a 3-colorable graph with 4 colors. More generally, for every $k\geq 3$ it is $\NP$-hard to color a $k$-chromatic graph with at most $k+2\left\lfloor\frac{k}{3}\right\rfloor-1$ colors.
\end{thm}

We now translate this theorem into a statement about Gr\"obner bases. Having additional colors to work with allows us to ignore certain vertices of our graph and color these later using our extra colors. Algebraically, this corresponds to ignoring certain variables and computing a Gr\"obner basis for the partial coloring ideal on the remaining variables.

\begin{defi}
Given a set of polynomials $\F\subseteq\mathbb{K}[x_1, \ldots, x_n]$, we say that a subset $X$ of the variables $x_1,\ldots,x_n$ is {\em independent} on $\F$ if no two variables in $X$ appear together in any element of $\F$.
\end{defi}

Clearly independent sets in our coloring ideal generators correspond to independent sets of vertices of the graph.

\begin{defi}
Define the \emph{strong $c$-partial Gr\"obner} problem as follows. Given as input, a set $\F$ of polynomials on a set $X$ of variables, output the following:
\begin{itemize}
\item disjoint $X_1,\ldots,X_b\subseteq X$, such that $b\le c$ and each $X_i$ is an independent set of variables,
\item $X'\subseteq X$, where $X'=X\backslash\left(\bigcup_i X_i\right)$ (i.e., we have taken away at most $c$ independent sets of variables),
\item $\F'\subseteq \F$ such that $\F'$ consists of all polynomials in $\F$ involving only variables in $X'$,
\item a Gr\"obner basis for $\lz\F'\rz$ over $X'$ (where the monomial order on $X$ is restricted to a monomial order on $X'$).
\end{itemize}
\end{defi}

\begin{thm}
\label{thm:stronggroebner}
Suppose that we are working over a polynomial ring $\K[x_1,\ldots,x_n]$ under some elimination order on the variables (such as lexicographic order). Assume furthermore that $\K$ is either a finite field or the field of rational numbers.

Let $k\ge 3$ be an integer, and set $c=2\left\lfloor\frac{k}{3}\right\rfloor-1$. Unless $\text{P}=\NP$, there is no polynomial-time algorithm $\A$ that solves the strong $c$-partial Gr\"obner problem (even if we restrict to sets of polynomials of degree at most $k$).
\end{thm}

The following lemma will be useful in our proof.



\begin{lemma} \label{FindSolution}
Suppose that we are given a Gr\"obner basis $\G$ for the $k$-coloring ideal $\I_G$ of a graph $G$, with respect to a given elimination order. Assuming the variety $\mathcal{V}(\I_G)$ is non-empty, there is an algorithm that finds some solution $x\in\mathcal{V}(\I_G)$ in time polynomial in the encoding length of $\G$ and $k$, and therefore identifies a $k$-coloring of $G$.
\end{lemma}

\begin{proof}
We assume that we are able to find the roots of univariate polynomials quickly (to any desired level of accuracy). Because the roots of the system are $k$th roots of unity, when we solve the first univariate polynomial (on the last variable in the order) we have only $k$ choices. As we back substitute at each new polynomial in $\G$, we have only $k$ possible solutions again, and may therefore find a solution in polynomial time. The Elimination Theorem guarantees that each partial solution can be extended in this manner to a complete solution. (See \cite[Chap. 3]{coxetal}.)
\end{proof}

\begin{proof}[Proof of Theorem \ref{thm:stronggroebner}]
The proof is by contradiction. Let $G=(V, E)$ be a $k$-colorable graph and assume such a polynomial-time algorithm $\A$ exists. We will give a method for producing a proper $(k+c)$-coloring of $G$.  This contradicts Theorem \ref{ApproxHard} under the assumption that $\text{P}\neq\NP$, as mentioned above.

Let us apply the algorithm $\A$ to our coloring polynomials $\F_G$ for the graph $G$, giving us a Gr\"obner basis $\G$. Note that the input consists of $|V|+|E|$ polynomials with degree $\leq k$ and length $\leq k$. Thus, $\F_G$ has polynomial size in $k$ and the encoding length of $G$, and by assumption $\A$ terminates in time which is polynomial in both of these quantities.

Observe that the variables in $\F_G$ correspond to vertices of $G$, and an independent set of variables corresponds to an independent set of vertices. Assume that the independent sets of variables which were ignored by $\A$ are $X_1,X_2,\ldots,X_b$ for $b\le c$.  Let $I_1,I_2,\ldots,I_b$ be the corresponding independent sets of vertices.  The Gr\"obner basis $\G$ corresponds to proper $k$-colorings of $G'=G\backslash \left(\cup_i I_i\right)$.  Therefore, Lemma \ref{FindSolution} implies that we can identify some proper coloring of $G'$ using the colors $1,\ldots,k$. Note that in order to apply this lemma, we must be working with an elimination order over our restricted polynomial ring; this is true since the restriction of an elimination order to a smaller set of variables is also an elimination order.

Now color the independent sets $I_1,\ldots,I_b$ in the colors $k+1,\ldots,k+b$. This gives us a proper coloring of $G$ using at most $$k+b\le k+c=k+2\left\lfloor\frac{k}{3}\right\rfloor-1$$colors. By Theorem \ref{ApproxHard}, this is impossible to construct in polynomial time, giving us a contradiction, as desired.
\end{proof}

Theorem \ref{thm:stronggroebner} demonstrates how results on coloring graphs translate effectively to results on Gr\"obner bases.  For reference, a weaker result may be proven without recourse to the full power of the coloring ideal.

\begin{defi}
Define the \emph{weak $c$-partial Gr\"obner} problem as follows. Given, as input, a set $\F$ of polynomials on a set $X$ of variables, output the following:
\begin{itemize}
\item $X'\subseteq X$ such that $|X'|\ge |X|-c$,
\item $\F'\subseteq \F$ such that $\F'$ consists of all polynomials in $\F$ involving only variables in $X'$,
\item a Gr\"obner basis for $\lz\F'\rz$ over $X'$ (where the monomial order on $X$ is restricted to a monomial order on $X'$).
\end{itemize}
\end{defi}

\begin{thm}
\label{thm:weakgroebner}
For constant $c$, there is no polynomial-time algorithm to solve the weak $c$-partial Gr\"obner problem, unless P$=$NP. This holds even if we restrict to sets of polynomials of degree at most 3.
\end{thm}

Our proof will use the following lemma.

\begin{lemma}
\label{lemma:disjoint}
Suppose that $\F_1,\F_2,\ldots,\F_m$ are sets of polynomials on disjoint sets of variables (that is, no variable appears both in a polynomial of $\F_i$ and in a polynomial of $\F_j$).  Then, the reduced Gr\"obner basis of $\lz\cup_i\F_i\rz$ is the union of the reduced Gr\"obner bases for the individual $\lz\F_i\rz$.
\end{lemma}
\begin{proof}
Let $\G_i$ be the reduced Gr\"obner bases for the $\lz\F_i\rz$, and set $\G:=\cup_i \G_i$. For a set $S$ of polynomials, we use $\mathcal{L}(S)$ to denote the ideal generated by the leading terms of elements of $S$. Observe that $$\mathcal{L}\left(\lz\cup_i \F_i\rz\right)=\mathcal{L}\left(\sum_i{\lz\F_i\rz}\right)=\sum_i{\mathcal{L}\left(\lz\F_i\rz\right)}=\sum_i{\mathcal{L}(\G_i)}=\mathcal{L}\left(\cup_i \G_i\right)=\mathcal{L}(\G)~.$$ Hence, $\G$ is a Gr\"obner basis of $\cup_i \lz\F_i\rz$.

To see that it is the (unique) reduced Gr\"obner basis, note that for all $i\in\{1, \ldots, m\}$, no term of an element of $\G_i$ is divided by a leading term of an element of $\G_j$, where $i\neq j$, and since $\G_i$ is reduced, the same holds for leading terms in $\G_i$. This suffices to show that $\G$ is reduced.
\end{proof}

\begin{proof}[Proof of Theorem \ref{thm:weakgroebner}.]
Suppose that there exists an algorithm $\A$ for $c$-partial Gr\"obner that runs in time at most $p(s)$, where $s$ is the size of the input.  Let $\F$ be a system of polynomials in $\K[x_1,\ldots,x_n]$ with input size $s$, such that the degree of every polynomial in $\F$ is at most 3. We show how to use $\A$ to compute a Gr\"obner basis for $\lz\F\rz$ in polynomial time, which will lead to a contradiction.

Construct copies $\F_1,\F_2,\ldots,\F_{c+1}$ of $\F$ on disjoint sets of variables, so that $\F_i$ includes polynomials over the variables $x_{i,1},x_{i,2},\ldots,x_{i,n}$.  The size of $\cup_i \F_i$ is obviously $(c+1)s$.  Now run $\A$ on $\cup_i \F_i$, removing at most $c$ variables from $\cup_i \F_i$. In the process, we remove certain polynomials from $\F_i$ to yield sets $\F'_i$ of polynomials.  The output of $\A$ is a Gr\"obner basis $\G$ for $\left\lz\cup_i \F'_i\right\rz$.

Now, since there are $c+1$ disjoint sets of variables $\{x_{i,1},x_{i,2},\ldots,x_{i,n}\}$, there must exist at least one $i$ such that we have not removed any variable in $\{x_{i,1},x_{i,2},\ldots,x_{i,n}\}$.  For this value of $i$, we have $\F'_i=\F_i$.  Transforming $\G$ to a reduced Gr\"obner basis is routine and can be performed in polynomial time.  Applying Lemma \ref{lemma:disjoint}, we see that the restriction of $\G$ to $\{x_{i,1},x_{i,2},\ldots,x_{i,n}\}$ gives a reduced Gr\"obner basis for $\F'_i=\F_i$. This immediately gives us a reduced Gr\"obner basis for $\lz\F\rz$.

Observe that $(c+1)s$ is the size of our input $\cup_i \F_i$ to $\A$.  Therefore, the time required by our algorithm is at most $p((c+1)s)\le (c+1)^{\deg(p)} p(s)$.  Since $\F$ was chosen arbitrarily, this implies that for every family of polynomials of input size $s$, a Gr\"obner basis can be found in polynomial time at most $(c+1)^{\deg(p)} p(s)$.  However, since 3-coloring is NP-hard, the general problem of finding a Gr\"obner basis cannot be performed in polynomial time, even if we assume that every $f\in \F$ has degree at most 3.  Thus we have a contradiction, and conclude that the algorithm $\A$ cannot exist.
\end{proof}

Comparing Theorems \ref{thm:stronggroebner} and \ref{thm:weakgroebner}, we see that the latter allows us to remove only a constant number of individual variables, not a constant number of independent sets. Furthermore, the set of polynomials constructed in Theorem \ref{thm:weakgroebner} is \emph{disconnected}, according to the following Definition \ref{defi:connected}, while the set of polynomials constructed in Theorem \ref{thm:stronggroebner} is \emph{connected}.  It appears more natural to consider connected sets of polynomials, which occur in many applications.

\begin{defi}
\label{defi:connected}
We say that a set $\F$ of polynomials is \emph{disconnected} if we can partition $\F$ into $\F_1,\F_2$ such that the variables for $\F_1$ and $\F_2$ are disjoint. Otherwise, we say that $\F$ is \emph{connected}.
\end{defi}

\section{Gr\"obner Bases for Chordal Graphs}\label{sec:chordal}

Even though graph coloring is hard for general graphs, the problem can be solved in linear time for chordal graphs (see e.g.\ \cite{Golumbic}). We develop an efficient algorithm that takes advantage of the structure of chordal graphs to compute a Gr\"obner basis for the $k$-coloring ideal $\I_G$ of a given chordal graph $G$. The monomial order we consider will depend upon the choice of $G$.

Recall that a graph $G=(V, E)$ is {\em chordal} if every cycle of length more than 3 has a chord, or equivalently, every induced cycle in the graph has length 3. A vertex $v\in V$ is {\em simplicial} if its neighbors form a clique. A graph is \emph{recursively simplicial} if it contains a simplicial vertex $v$ such that $G[V\setminus\{v\}]$ is recursively simplicial (the graph on zero vertices is defined to be recursively simplicial). If $G$ is recursively simplicial, there is an ordering in which the vertices are removed such that each vertex is simplicial at the time of removal. This order is called a {\em perfect elimination ordering}.  Therefore, a recursively simplicial graph $G$ is constructed by adding vertices according to the reverse perfect elimination ordering, such that each vertex is simplicial when added.

\begin{prop}[\cite{FulkersonGross65}]
Let $G=(V, E)$ be a graph. Then $G$ is chordal if and only if it is recursively simplicial.
\end{prop}

For a graph $G=(V, E)$ and vertex $v\in V$, we use the notation $\mathcal{N}(v)=\{w\in V:(v, w)\in E(G)\}$ to denote the {\em neighborhood} of $v$ in $G$. If $U\subseteq V$ is a subset of the vertices which forms a clique in $G$, then $$G^{+U}:=\big(V\cup\{n+1\}, E\cup\{(j, n+1): j\in U\}\big)$$ is the graph obtained by adding a new vertex and connecting it to all $u\in U$. This operation is the inverse of deleting a simplicial vertex of $G$.

Our Gr\"obner basis algorithm will build up a chordal graph one vertex at a time, according to the reverse elimination order. Each newly added vertex will add a polynomial to the Gr\"obner basis. At any point, having constructed the graph $G'\subseteq G$, the set of polynomials added will form a Gr\"obner basis for the coloring ideal of $G'$. 

\subsection{Preliminaries}

Recall the following definitions. The {\em $k$th elementary symmetric polynomial} $\sigma_k(x_1, \ldots, x_n)$ over $n$ variables is given by $$\sigma_k(x_1, \ldots, x_n):=\sum_{1\leq j_1<\cdots<j_k\leq n}{x_{j_1}\cdots x_{j_k}}\ \ .$$
The {\em $k$th complete homogeneous symmetric polynomial} $S_k(x_1, \ldots, x_n)$ over $n$ variables is given by
$$S_k(x_1, \ldots, x_n):=\sum_{1\leq j_1\leq\cdots\leq j_k\leq n}{x_{j_1}\cdots x_{j_k}}\ \ .$$Note that both polynomials are degree-$k$-homogeneous, but the monomials of $\sigma_k$ are by definition square-free, while $S_k$ can contain higher powers of a variable.

\begin{lemma}
\label{lemma:AugmentingPolynomial}
For a positive integer $k$ , let $\z_1,\z_2,\ldots,\z_k$ be the $k$th roots of unity in some order.  Then, for every $k>r$, we have $S_{k-r}(\z_1,\z_2,\ldots,\z_r,x)=(x-\z_{r+1})(x-\z_{r+2})\cdots (x-\z_k)$.
\end{lemma}
\begin{proof}
It suffices to prove $$S_{k-r}(\z_1,\z_2,\ldots,\z_r,x)\cdot(x-\z_1)\cdots(x-\z_r)=x^k-1\ \ .$$Consider the degree $d$-homogeneous polynomial $\sigma_i(\z_1, \ldots, \z_r)S_{d-i}(\z_1, \ldots, \z_r)$. For every monomial $x^\alpha$ with $|\alpha|=d$ and $\supp(\alpha)=m$ (the number of non-zero elements in $\alpha$ equals $m$), its coefficient is the number of square-free factors of degree $i$, that is, $m \choose i$. Summing up these coefficients over $d$ with alternating signs gives that the coefficient of $x^\alpha$ in $$\sum_{i = 0}^d (-1)^{d-i} \sigma_i(\z_1, \ldots, \z_r)S_{d-i}(\z_1, \ldots, \z_r)$$ equals $$\sum_{i = 0}^m(-1)^{d-i}{m\choose i} = 0\ \ .$$

Therefore, $$\sum_{i = 0}^d (-1)^{d-i} \sigma_i(\z_1, \ldots, \z_r)S_{d-i}(\z_1, \ldots, \z_r)=0\ \forall d\in\{0, \ldots, k-1\}\ \ .$$Now, since $\z_1,\ldots,\z_k$ are the roots of unity, we know that $$\sum_{i = 0}^d \sigma_i(\z_1, \ldots, \z_r)\sigma_{d-i}(\z_{r+1}, \ldots, \z_k)=\sigma_d(\z_1,\ldots,\z_k)=0\ \forall d\in\{1, \ldots, k-1\}\ \ .$$

We now have identical recursions for $S_d(\z_1, \ldots, \z_r)$ and $(-1)^{d} \sigma_d ( \z_{r+1}, \ldots, \z_k)$.   In the base case, $S_0(\z_1, \ldots, \z_r)=1=(-1)^0\sigma_0(\z_{r+1},\ldots,\z_k)$.  We conclude that for all $d$, $$S_d(  \z_1, \ldots, \z_r) = (-1)^{d} \sigma_d(\z_{r+1}, \ldots, \z_k)~.$$Therefore,
\begin{eqnarray*}
S_{k-r}(\z_1, \ldots, \z_r, x)\cdot\prod_{i=1}^r(x-\z_i)&=&\sum_{d = 0}^{k-r}S_d(\z_1, \ldots, \z_r)x^{k-r-d}\cdot\prod_{i=1}^r(x-\z_i)\\&=&\sum_{d = 0}^{k-r}(-1)^d\sigma_d(\z_{r+1}, \ldots, \z_k)x^{k-r-d}\cdot\prod_{i=1}^r(x-\z_i)\\&=&\prod_{i=r+1}^k(x-\z_i)\cdot\prod_{i=1}^r(x-\z_i)\\&=&x^k-1\ \ ,
\end{eqnarray*}
as desired.
\end{proof}

\subsection{The algorithm}

Now we are ready to present the Gr\"obner basis algorithm {\bf BuildGr\"obnerBasis}. Our algorithm successively tests vertices for simpliciality in order to obtain a perfect elimination order. (It is certainly possible to achieve faster running time by refining this procedure.) At the same time, we add new polynomials to a set $\G$. At termination, $\G$ is a Gr\"obner basis for $\I_G$ with respect to the lexicographic order in which vertices are ordered accorded to a perfect elimination order. The existence of this algorithm was first conjectured by experimental work of Pernpeintner \cite{michael_thesis}.

For a clique $U=\{u_1,u_2,\ldots,u_r\}$ and vertex $v$ in our graph, we will use the notation $S_{k-r}(U,v)$ to denote the polynomial $S_{k-r}(x_{u_1},x_{u_2},\ldots,x_{u_r},x_v)$.

\begin{minipage}{0.45\linewidth}
\textbf{Input:} A graph $G$\newline
\textbf{Output:} A Gr\"obner basis for $\I_G$
\begin{algorithmic}
\Function{BuildGr\"obnerBasis}{$G$}
	\State $G_n\gets G$
	\State $\G \gets \emptyset$
	\ForAll{$i\in\{n-1, \ldots, 1\}$}
		\ForAll{$v\in V_{i+1}$}
			\If {\Call{IsSimplicial}{$v$}} 
				\State $v_i \gets v$
				\State $U_i \gets \mathcal{N}(v)$
				\State $G_i \gets G_{i+1}-v$
				\State $\G \gets \G\cup\{S_{k-|U_i|}(U_i, v_i)\}$
				\State \ExitFor
			\EndIf
		\EndFor
	\EndFor
	\State \Return $\G$
\EndFunction
\end{algorithmic}
\end{minipage}
\begin{minipage}{0.45\linewidth}
\textbf{Input:} A vertex $v$ of the graph $G$\newline
\textbf{Output:} Whether or not $v$ is simplicial
\begin{algorithmic}
\Function{IsSimplicial}{$v$}
	\State $d\gets \deg(v)$
	\ForAll{$w\in\mathcal{N}(v)$}
		\If {$|\mathcal{N}(v)\cap\mathcal{N}(w)|< d-1$} 
			\State \Return false
		\EndIf
	\EndFor
	\State \Return true
\EndFunction
\end{algorithmic}
\vspace{93pt}
\end{minipage}


\subsection{Correctness}
\begin{thm}\label{lemma:ColoringIdealsRadical}
Let $G$ be a graph. Then $\I_G$ is a radical ideal.
\end{thm}
\begin{proof}
For every $i\in\{1, \ldots, n\}$, we have $\nu_i(x)=x_i^k-1\in\I_G\cap\mathbb{K}[x_i]$ by definition. Since $\mathbb{K}$ is algebraically closed and therefore $\nu_i'(x)=k\cdot x_i^{k-1}\ \Longrightarrow\ \gcd(\nu_i, \nu_i')=1$, we can apply Seidenberg's Lemma (\cite{KreuzerRobbiano}, Proposition 3.7.15), which gives the claim.
\end{proof}

\begin{prop}[\cite{coxetal}]\label{prop:RelativelyPrimeGB}
Let $P\subset\mathbb{K}[x_1, \ldots, x_n]$ be a finite set, and let $p_1, p_2\in P$ be relatively prime. 
Then $\text{S-pair}(p_1, p_2)\rightarrow_P 0$.
\end{prop}
Recall that $v_i \in \I_G$, and $\eta_{ij} \in \I_G$ are the vertex and edge polynomials, respectively.
\begin{lemma}\label{lemma:Correctness1}
Let $G$ be a graph on $n$ vertices, and let $\succ$ be a term order. Let $U=\{u_1, \ldots, u_r\}$ be an $r$-clique in $G$, and choose a Gr\"obner basis $\G$ of $\I_G$. Set $p=S_{k-r}(x_{u_1}, \ldots, x_{u_r}, x_{n+1})$. Then,\vspace{-3pt}$$\lz \G, p\rz=\lz \G, \nu_{n+1}, \eta_{u_1, n+1}, \ldots, \eta_{u_r, n+1}\rz=\I_{G^{+U}}\ \ .$$
\end{lemma}
\begin{proof}
We show that $\lz \G, p\rz$ is a radical ideal, and that both ideals generate the same variety. Then the claim follows from the bijection between varieties and radical ideals (\cite{coxetal}, Chapter 4, \S2, Theorem 7).

Consider some setting of the variables $x_{u_1},\ldots,x_{u_r}$ to distinct $k$th roots of unity $\z_1,\ldots,\z_r$, and suppose that $\z_{r+1},\ldots,\z_k$ are the other $k$th roots of unity, in some order.  By Lemma \ref{lemma:AugmentingPolynomial}, we have $p=\prod_{i=r+1}^k(x_{n+1}-\z_i)$. This implies that $p(x_{u_1},x_{u_2},\ldots,x_{u_r},x_{n+1})$ is a square-free polynomial so $\lz p\rz$ is a radical ideal. The ideal $\lz\G\rz$ is also radical, since it is the coloring ideal of a graph (Lemma~\ref{lemma:ColoringIdealsRadical}). But then $$\rad(\lz \G, p\rz)=\rad(\lz\G\rz\cap\lz p\rz)=\rad(\lz\G\rz)\cap\rad(\lz p\rz)=\lz\G\rz\cap\lz p\rz=\lz\G, p\rz$$ as claimed. The second equality is \cite{coxetal}, Chapter 4, \S3, Proposition 16.

Now consider $\xx=(x_1, \ldots, x_{n+1})\in\V(\lz\G, p\rz)$. Since $u_1,\ldots,u_r$ form a clique, we know that $x_{u_i}$ are distinct $k$th roots of unity. Then, by Lemma \ref{lemma:AugmentingPolynomial}, $x_{n+1}$ is a $k$th root of unity, and so $\nu_{n+1}=0$. Moreover, $x_{n+1}\neq x_{u_i}\ \forall\ i\in\{1, \ldots, r\}$, which implies that $\eta_{u_i,n+1}=0$. We conclude that $\xx\in\V(\I_{G^{+U}})$.

Conversely, consider $\xx=(x_1, \ldots, x_{n+1})\in\V(\I_{G^{+U}})$. The generator polynomials $\nu_1, \ldots, \nu_r, \nu_{n+1}$ and $\eta_{u_1, n+1}, \ldots, \eta_{u_r, n+1}$ ensure that $x_{u_1}, \ldots, x_{u_r},x_{n+1}$ are distinct $k$th roots of unity. Hence $p(\xx)=0$ and $\xx\in\V(\lz\G, p\rz)$, completing our proof.
\end{proof}

\begin{lemma}\label{lemma:Correctness2}
For every Gr\"obner basis $\G$ of $\I_G$ with respect to $\succ$, $\G\cup\{p\}$ is a Gr\"obner basis of $\I_{G^{+U}}$ with respect to an extended term order $\succ'$, where $p$ is again defined as in Lemma \ref{lemma:Correctness1}.
\end{lemma}
\begin{proof}
Lemma \ref{lemma:Correctness1} shows that $\lz\G, p\rz=\I_{G^{+U}}.$ Hence, it is left to show that all $S$-polynomials in $\G\cup \{p\}$ reduce to $0$. We only have to consider $S$-pairs that involve the new polynomial $p$. \\
By definition of $\succ'$, we have that $\LM_{\succ'}(p)=x_{n+1}^{k-r}$, which is relatively prime to all $g\in\G$, since $x_{n+1}$ does not appear in this basis. Therefore, $$\text{S-pair}(g, p)\rightarrow_{\G\cup\{p\}}0\ \ \forall g\in\G$$ by Lemma \ref{prop:RelativelyPrimeGB}. This is sufficient for $\G':=\G\cup\{p\}$ to be a Gr\"obner Basis.
\end{proof}

\begin{thm}
Upon termination of \textsc{BuildGr\"obnerBasis}$(G)$, the set $\G$ is a Gr\"obner basis for $\I_G$ under the $\Lex$ order, where the vertices are ordered in the perfect elimination order that was established in the algorithm.
\end{thm}
\begin{proof}
Note that $\{p_1:=\nu_n\}$ is a Gr\"obner basis for $G_1$. By Lemma \ref{lemma:Correctness2}, this basis can be extended in $n-1$ steps by adding $p_i$ as constructed in the algorithm. Therefore, $\G=\{p_1, \ldots, p_n\}$ is a Gr\"obner basis of $G_n=G$ with respect to the extended vertex order, which concludes the proof. 
\end{proof}

\subsection{Remarks}
As we have seen above, exactly one polynomial is added to $\G$ for every vertex of $G$. But what is the degree and length of these polynomials?

From the definition of $p:=S_k(x_1, \ldots, x_n)$, we see that $\len(p)={{k+n-1}\choose{n-1}}$ and $\deg(p)=k$. Therefore, we add polynomials $S_i$ with $\len(S_i)={k\choose{|U_i|}}$ and $\deg(S_i)=k-|U_i|$. Note that, for a fixed number $k$ of colors, both numbers are polynomials. 

The function \textsc{IsSimplicial} consists of an outer loop with exactly $n$ iterations, in each of which the intersection of two subsets of $V$ is formed. Such an intersection can be computed in linear time, therefore the function runs in time $O(n^2)$.

In the main function \textsc{BuildGr\"obnerBasis}, the two nested \textbf{for}-loops are traversed $O(n)$ times each, and every time \textsc{IsSimplicial} is called. The main part of the \textbf{if}-case is the assignment of $\G$. If $r=|U_i|$, then building the polynomial $S_{k-|U_i|}(U_i, v_i)$ takes $(k-r)\cdot{k\choose r}$ steps, which is clearly in $O(kn^k)$. The remaining statements can be neglected, since they have running time $O(n^2)$. Finally, putting the pieces together, we obtain a total running time of $$O(kn^{k+2})\ \ ,$$which is polynomial in $n$ for fixed $k$.

It is evident that our implementation is not optimal with respect to running time. For instance, finding a simplicial vertex can be done in linear time \cite{MCSalg}, and there is even a linear-time procedure that establishes a perfect elimination order on $G$. Nevertheless, our algorithm shows that finding the Gr\"obner basis for a chordal graph is polynomial-time solvable, and it describes explicitly the structure of this basis.

What happens in the process of the algorithm if $G$ is not $k$-colorable? Intuitively, we would expect the constant polynomial $1$ to appear somewhere in the set $B$.  This can be shown formally: Assume that $\chi(G)=\chi>k$, and we try to find a Gr\"obner basis for the $k$-coloring ideal of $G$. Since $G$ is chordal, it is also perfect, and thus has a $\chi$-clique $\{v_1, \ldots, v_\chi\}$. We assume without loss of generality that these vertices are ordered ascendingly with respect to the perfect elimination order from the algorithm. 

In the step, where $v_{k+1}$ is removed from the graph, we have $\{v_1, \ldots, v_k\}\subseteq \mathcal{N}(v_{k+1})$, and therefore, we add the complete polynomial of degree 0 $$S_{k-k}(x_{v_1}, \ldots, x_{v_k}, x_{v_{k+1}})=1\ \ .$$
Hence, \textsc{BuildGr\"obnerBasis} detects non-$k$-colorability on the fly. This observation suggests the following simple improvement on the algorithm: If we find a simplicial vertex of degree $\geq k$, then we can stop immediately and return the trivial Gr\"obner basis $B=\{1\}$. On the other hand, we can be sure that if there is no such forbidden vertex, then $G$ is $k$-colorable.


\section{Acknowledgements}

This research is based upon work supported by the National Science Foundation Grant No.~DMS-1321794 and the NSF Graduate Research Fellowship under Grant No.~1122374.  We are very grateful to the AMS Mathematical Research Communities program, and especially Ellen J.~Maycock, for their support of this project. The authors wish to thank Hannah Alpert for her extremely helpful thoughts. We are also grateful to the Simons Institute and would like to thank Agnes Szanto and Pablo Parrilo for their constructive comments.


\end{document}